%% file: arXiv.tex
\documentclass[11pt]{article}
\RequirePackage{algorithm,algorithmic,graphicx,amssymb,epsf,epic,url,pifont,complexity}
\RequirePackage{fullpage}

\newlength{\bibitemsep}
\newlength{\bibparskip}
\let\oldthebibliography\thebibliography
\renewcommand\thebibliography[1]{%
  \oldthebibliography{#1}%
  \setlength{\parskip}{\bibitemsep}%
  \setlength{\itemsep}{\bibparskip}%
}

\input{latexmac}

\denseformat

\long\def\commabs #1\commabsend{}

\begin{document}

\title{Local Algorithms for Bounded Degree Sparsifiers in Sparse Graphs} 
\author{Shay Solomon \thanks{Tel Aviv University. This work was carried out while the author was a postdoc at IBM Research and supported by the IBM Herman Goldstine Postdoctoral Fellowship. E-mail: {\tt solo.shay@gmail.com}.}}

\date{\empty}

\begin{titlepage}
\def\thepage{}
\maketitle

\begin{abstract}
In graph sparsification, the goal has almost always been of \emph{global} nature: compress a graph into a smaller subgraph (\emph{sparsifier}) that maintains certain features of the original graph.
Algorithms can then run on the sparsifier, which in many cases leads to improvements in the overall runtime and memory.
This paper studies sparsifiers that have bounded (maximum) degree, and are thus \emph{locally} sparse, aiming to improve local measures of runtime and memory.
To improve those local measures, it is important to be able to compute such sparsifiers \emph{locally}.

We initiate the study of local algorithms for bounded degree sparsifiers in unweighted sparse graphs,
focusing  on the problems of vertex cover, matching, and independent set.
Let $\eps > 0$ be a slack parameter and $\alpha \ge 1$ be a density parameter.
We devise local algorithms for computing:
\begin{enumerate}
\item  A $(1+\eps)$-vertex cover sparsifier of degree $O(\alpha / \eps)$, for any graph of \emph{arboricity} $\alpha$.\footnote{In a graph of arboricity $\alpha$ the average degree of any induced subgraph is at most $2\alpha$.}
\item A $(1+\eps)$-maximum matching sparsifier and also a $(1+\eps)$-maximal matching sparsifier of degree $O(\alpha / \eps)$, for any graph of arboricity $\alpha$.
 \item A $(1+\eps)$-independent set sparsifier of degree $O(\alpha^2 / \eps)$, for any graph of average degree $\alpha$.
\end{enumerate}
Our algorithms require only a single communication round in the standard message passing models of distributed computing,
and moreover, they  can be simulated locally in a trivial way.
As an immediate application we can extend   results from distributed computing and local computation algorithms that apply to graphs of degree bounded by $d$ to graphs of arboricity $O(d / \eps)$ or average degree $O(d^2 / \eps)$,   at the expense of increasing the approximation guarantee by a factor of $(1+\eps)$. In particular, we can extend the plethora of recent local computation algorithms for approximate maximum and maximal matching from bounded degree graphs to bounded arboricity graphs with a negligible loss in the approximation guarantee.

The inherently local behavior of our algorithms can be used to amplify the approximation guarantee of any sparsifier in time roughly linear in its size,
which has immediate applications in the area of dynamic graph algorithms. In particular, the state-of-the-art algorithm for
maintaining $(2-\eps)$-vertex cover (VC) is at least linear in the graph size, even in dynamic forests. We provide a reduction from the dynamic to the static case,
showing that if a $t$-VC can be computed from scratch in time $T(n)$ in any (sub)family of graphs with arboricity bounded by $\alpha$, for an arbitrary $t \ge 1$,
then a $(t+\eps)$-VC can be maintained with update time $\frac{T(n)}{O((n / \alpha) \cdot \eps^2)}$, for any $\eps > 0$. For planar graphs this yields an algorithm for maintaining
a $(1+\eps)$-VC with constant update time for any constant $\eps > 0$.
\end{abstract}
\end{titlepage}

\section{Introduction}
Graph sparsification has been extensively studied for many years, and is subject to increasingly growing interest due to the rapidly growing necessity of dealing with huge-sized graphs.
Given such a graph $G = (V,E)$, we would like to \emph{compress} $G$ into a subgraph $H$ of much smaller size that maintains certain features of $G$, such as distances, cuts or flows. Algorithms can then run on the compressed subgraph $H$, sometimes called \emph{sparsifier}, rather than the original graph $G$, which may save significantly on important resources such as the overall runtime and memory of the algorithm, often at the expense of approximate rather than exact solutions or worse approximation guarantees.
The most common type of sparsifiers are \emph{edge sparsifiers}, such as graph spanners \cite{PU89}
and cut or spectral sparsifiers \cite{BK96,ST04}, which span the original vertex set using a small number of edges.
Another well-studied type of sparsifiers are \emph{vertex sparsifiers}, such as flow or cut sparsifiers \cite{HKNR98,Moitra09,LM10}, which should span a small number of designated vertices.

The basic goal in this area has almost always been of \emph{global} nature, i.e., of minimizing the overall size of the sparsifier and the overall time needed for computing it.
One of the exceptions is in the area of spanners, where researchers have studied spanners of bounded (maximum) \emph{degree}.
While a sparse spanner has a low average degree, and is thus globally sparse,
a bounded degree spanner has a low maximum degree, and is thus locally sparse.
Although bounded degree spanners have been little studied thus far in general graphs \cite{CDK12,CD14},
they have been studied   extensively in Euclidean low-dimensional spaces, see e.g., \cite{DHN93,ADMSS95,DN97,GLN02,ES15}).
The spanner degree often determines local memory constraints when using spanners to construct network synchronizers \cite{PU89} and efficient broadcast protocols \cite{ABP90,ABP91}.
In compact routing schemes (e.g., \cite{TZ01,Chechik13}), the use of low degree spanners may enable the routing tables to be of small size.
Moreover, viewing vertices as processors, in many applications the degree of a processor
represents its \emph{load}, hence a low degree spanner guarantees that the load on all the processors in the network will be low.

This paper studies sparsifiers from a \emph{local} perspective, aiming to improve local measures of runtime and memory.
To improve those local measures, it is important to be able to compute such sparsifiers \emph{locally}, in a manner to be defined shortly.
We initiate the study of local algorithms for \emph{bounded  degree} sparsifiers in unweighted \emph{sparse} graphs.
The graphs that we consider are sparse either globally, i.e., of bounded average degree, or uniformly, i.e., of bounded \emph{arboricity}, whence the average degree of any induced subgraph is bounded. In   sparse graphs some vertices may have   large degrees, as with the $n$-star graph.
Our basic goal is to compute \emph{locally} a sparsifier $H$ for the original graph $G = (V,E)$, whose maximum degree is bounded in terms of the density of $G$ and some slack parameter $\eps > 0$, and which approximately preserves a certain property or feature of the original graph; the sparsifier would ideally be a subgraph of $G$, but this cannot always be achieved.
Algorithms, particularly local ones, can then run on the bounded degree sparsifier $H$ rather than on the original graph $G$, which may save significantly on local resources of runtime and memory.
For concreteness, we focus  on the following combinatorial optimization problems: (approximate) minimum vertex cover (VC),
 maximum and maximal matching, and maximum independent set (IS). It would be only natural to extend the study of bounded degree sparsifiers to other fundamental problems.


For the maximum matching problem, a $(1+\eps)$-sparsifier for $G$ is a \emph{subgraph} $H = (V',E')$ of $G$, with $V' \subseteq V, E' \subseteq E$, such that the maximum matching size of $H$ is within a factor of $1+\eps$ from that of $G$; thus a $(1+\eps)$-(approximate) maximum matching for $H$ is a $(1+O(\eps))$-maximum matching for $G$.
For the maximal matching problem the definition is similar; see Section \ref{prel}.
We need to be more careful with the definitions of sparsifier for the minimum VC and maximum IS problems,
since a VC (respectively, IS) for a subgraph $H$ of $G$ may not be a \emph{valid} VC (resp., IS) for the entire graph $G$.
Note that we are concerned with the validity of solutions obtained by the sparsifier rather than the approximations that they provide.
Consequently, a sparsifier in these cases will not be simply a subgraph $H$ of $G$, but rather a pair $(H,V')$,
where $H$ is a subgraph of $G$ and $V'$ is a vertex set of $V$, hereafter the \emph{validating set} of the sparsifier, such that for any VC (resp., IS) for $H$,
adding (resp., removing) the validating set $V'$ to (resp., from) it provides a valid VC (resp., IS) for $G$;
the role of the validating set is to translate the solution obtained by the sparsifier into a valid solution for $G$.
We say that $(H,V')$ is a $(1+\eps)$-sparsifier for $G$ if for any $(1+\eps)$-(approximate minimum) VC (resp., (approximate maximum) IS) for $H$,
denoted by $C_H$ (resp., $I_H$), the set $C_H \cup V'$ is a valid $(1+O(\eps))$-VC (resp., $I_H \setminus V'$ is a $(1+O(\eps))$-IS) for $G$.


Given any  $\eps > 0$ and any density parameter $\alpha \ge 1$, we devise \emph{local} algorithms for computing:
\begin{enumerate}  
\item
A $(1+\eps)$-VC sparsifier of degree $O(\alpha / \eps)$, for any graph of \emph{arboricity} bounded by $\alpha$.
\item A $(1+\eps)$-maximum matching sparsifier and also a $(1+\eps)$-maximal matching sparsifier of degree $~~~~~O(\alpha / \eps)$, for any graph of arboricity bounded by $\alpha$.
\item A $(1+\eps)$-IS sparsifier of degree $O(\alpha^2 / \eps)$, for any graph of average degree bounded by $\alpha$.
\end{enumerate}

Aiming at enhancing the applicability and usefulness of our sparsifiers, we adhere to a strict notion of locality: For any vertex $v$, we
want to be able to compute the adjacent edges of $v$ that belong to the sparsifier by \emph{probing} only $v$ and a small number (bounded by the degree of the sparsifier) of its neighbors, where the probing procedure is context-dependent.
In standard centralized settings such a procedure will simply examine the data structures of $v$ and those neighbors,
but in the message passing models of distributed computing, for example, the procedure may trigger the exchange of messages between $v$ and  those neighbors.
Also, we want to determine if a vertex $v$ belongs to the validating set of the sparsifier by probing only $v$.
The advantage in using such a strict notion of locality is three-fold, as summarized here and described in more detail later on:
\begin{enumerate}
\item In the rapidly growing area of \emph{local computation algorithms} (see, e.g., \cite{RTVX11,ARVX12,EMR14,EMR15}), a standard assumption is that the underlying graph has bounded degree.
This assumption is required since a local computation algorithm would typically probe all vertices inside a small-radius ball around the queried vertex/edge.
If the maximum degree is $\Delta$ and the ball radius is $r$, the probe complexity is bounded by $\Delta^{O(r)}$, and sometimes the total runtime and space will also be bounded by $\Delta^{O(r)}$.
Due to the local nature of our sparsification algorithms, we can restrict the probing procedure only to the sparsifier edges,
which directly enables us to extend known results from bounded degree graphs to uniformly sparse graphs.
\item In dynamic centralized graph algorithms, following an update of a vertex/edge, the update algorithm would typically scan all   neighbors of the updated vertices (and usually more than just those vertices),
either to obtain up-to-date information from the   data structures of those neighbors or to update that information.
Since the adversary may choose to focus its attention on few high degree vertices, this could lead to algorithms with a   poor update time.
Using our notion of locality, we show that the attention can be restricted to only few edges of the sparsifier,
which leads to   improvements in the update time.
\item In distributed networks, we can compute the sparsifier in a single communication round. Moreover, since each vertex communicates with only few of its neighbors, the   \emph{load}  on all vertices (or processors) throughout the sparsifier's computation is low.
    After the sparsifier has been computed, running on it distributed algorithms rather than on the original network may significantly reduce the total runtime of the algorithms and the load on the processors.
\end{enumerate}
In addition to the above applications, our sparsification algorithms can be used more broadly in computational models where there are local memory constraints,
such as the distributed communication model and the massively parallel computation (MPC) model, which is an abstraction of MapReduce-style frameworks (cf.\ \cite{CLMMOS17,ABBMS17}).
Another relevant model is the dynamic distributed model  (cf.\ \cite{PPS16,CHK16}), where some graph structure (e.g.,  matching) is to be maintained in a dynamically changing distributed network using low local memory at processors. 
\subsection{Our sparsifiers} \label{s11}
Perhaps the most important feature of our sparsification algorithms is their   simplicity, which is partly why they can be computed under such a strict notion of locality.

For any   $\Delta \ge 1$,
let $V^\Delta_{high}$ and $V^\Delta_{low}$ be the sets of vertices of degree $\ge \Delta$ and $< \Delta$, respectively.
When $\Delta$ is clear from the context, we may omit it from the superscript.
For any vertex set $V'$  in $G$,   denote by $G[V']$ the subgraph induced by $V'$. Define $G_{high} = G[V_{high}]$ and $G_{low} = G[V_{low}]$.

\begin{enumerate}
\item For the minimum VC problem, we take the pair $(G_{low},V_{high})$ as the $(1+\eps)$-sparsifier for $G$,
where $G_{low}$ is a subgraph of $G$ and $V_{high}$ is the validating set of the sparsifier.
It is clear that the degree of $G_{low}$ is at most $\Delta$, and moreover, for any VC for $G_{low}$, its union with $V_{high}$ is a valid VC for $G$.
In Section \ref{s32} we show that for any graph of arboricity $\alpha$, taking $\Delta = O(\alpha / \eps)$ guarantees the following: If $VC_{low}$ is a $(1+\eps)$-VC for $G_{low}$, then $VC_{low} \cup V_{high}$ is a $(1+O(\eps))$-VC for $G$.

\item For the maximum and maximal matching problems, a (subgraph) $(1+\eps)$-sparsifier $G_\Delta$ for $G$ with degree bounded by $\Delta$ can be obtained as follows:
Mark up to $\Delta$ arbitrary adjacent edges on every vertex $v$, and add to $G_\Delta$ all edges that are marked  by both endpoints.
It is clear that the degree of $G_\Delta$ is at most $\Delta$. (Note that if we took to $G_\Delta$ edges that are marked just once, the degree of $G_\Delta$ could explode.)
In Section \ref{s31} we show that for any graph of arboricity $\alpha$, taking $\Delta = O(\alpha / \eps)$ guarantees that the subgraph $G_\Delta$ is a $(1+\eps)$-sparsifier for $G$.


\item For the maximum IS problem, we take  $G_{low}$  as the $(1+\eps)$-sparsifier for $G$. (Although we may also use a validating set for the sparsifier, there is no need to do that here; thus in this case the IS sparsifier is a subgraph of $G$.)
It is clear that the degree of $G_{low}$ is at most $\Delta$, and moreover, any IS  for $G_{low}$ is a valid IS for $G$.
In Section \ref{s33} we show that for any graph of average degree $\alpha$, taking $\Delta = O(\alpha^2 / \eps)$ guarantees that
any $(1+\eps)$-IS for $G_{low}$ is   a $(1+O(\eps))$-IS for $G$.
\end{enumerate}
Note that our sparsifiers are obtained by essentially  ``ignoring'' the high degree vertices, where what is meant by ignoring is context-dependent.
For the minimum VC and maximum IS problems, we take all high degree vertices to the VC and
take none of them to the IS, respectively,
whereas for the maximum and maximal matching problems, we ignore all but at most $\Delta$ edges adjacent on any high degree vertex.
This approach of ignoring the high degree vertices  
can be viewed as a general paradigm, and it would be interesting to apply it to additional fundamental graph problems.

\subsection{Local computation algorithms} \label{s12}
The model of \emph{local computation algorithms} was introduced by Rubinfeld et al.\ \cite{RTVX11}, motivated by the fact that it is prohibitively expensive and sometimes   infeasible for an algorithm to read and process the entire input as well as to report the entire output, when dealing with massive data sets.
Local computation algorithms should answer \emph{queries} regarding global solutions to computational problems by examining only a small part of the input.
The goal is   to reach a global solution by performing local (sublinear time) computations on the input, and answer only regarding the queried part of the output.
If there are multiple possible solutions, the answers to all queries must be consistent with a single solution.
(More technical details on this model are given in Section \ref{prel}; see also \cite{RTVX11}.)
For the $(1+\eps)$-maximum matching problem, each query is an edge in the graph, and the algorithm needs to answer whether the queried edge belongs
to a $(1+\eps)$-maximum matching; note that the answers to all queries must be with respect to the same matching.
\cite{MV13} devised a randomized local computation algorithm for $(1+\eps)$-maximum matching with time and space complexities
$\poly(\log n) \cdot \exp(\Delta)$, where $\Delta$ is the maximum degree of the graph. This result was improved in  \cite{EMR14} to a deterministic algorithm with
time complexity $O(\log^* n) \cdot \exp(\Delta)$ and zero space complexity.
\cite{LRY17} devised a randomized algorithms with time and space complexities of
$\poly(\log n,\Delta)$.  \cite{Fischer17} obtained a deterministic algorithm for $(2+\eps)$-maximum matching with time complexity
$O(\log^* n) \cdot 2^{O(\Delta^2)}$.
(We ignore the dependencies on $\eps$ in the results of \cite{MV13,EMR14,LRY17,Fischer17}; in fact, in some of these results it is assumed that $\eps$ is constant.)

We can extend the results of \cite{MV13,EMR14,LRY17,Fischer17} from graphs of bounded degree to graphs of bounded arboricity.
Specifically, for any graph with arboricity bounded by $\alpha$, our matching sparsifier $G_\Delta$ has a degree bounded by $\Delta = O(\alpha/\eps)$.
We get this extension by exploiting the local nature of $G_\Delta$, and in particular,
the fact that for any vertex $v$, we can compute the adjacent edges of $v$ that belong to $G_\Delta$ by probing only $v$ and at most $\Delta$ of its neighbors.
Any $(1+\eps)$-maximum matching computed for the sparsifier provides a $(1+\eps)^2 = (1+O(\eps))$-maximum matching for the original graph,
thus there is only a negligible loss in the approximation guarantee.
Since $\Delta = O(\alpha/\eps)$, the smaller $\eps$ is, the larger the time and space complexities get. Nonetheless, as long as $\eps$ is not too small, the loss here is quite negligible too.
In this way reduce the problem of approximate maximum matching from graphs of arboricity bounded by $\alpha$ to graphs of degree bounded by $\approx \alpha$.  

In the same way we reduce the problems of approximate minimum VC and maximum IS from bounded arboricity graphs
and graphs of bounded average degree, respectively, to bounded degree graphs.
These reductions show that if and when local computation algorithms for these problems
are developed in bounded degree graphs (there are currently no such algorithms),
 they will immediately give rise to new algorithms in the respective wider families.
Moreover, this can be viewed as a general paradigm: By locally computing a sparsifier for a combinatorial optimization problem in some family of graphs,
we reduce the problem from that family to the family of bounded degree graphs, and the loss depends on the approximation guarantee of the sparsifier and on its degree.
\subsection{Dynamic centralized graph algorithms} \label{s13}
The problems of dynamically maintaining approximate minimum VC and maximum   matching have been intensively studied in recent years, see e.g.\
\cite{OR10,BGS11,NS13,GP13,BHI15,BS16,BHN16}.
The holy grail is for the approximation guarantee to approach 1 and for the (amortized or worst-case) update time   to be $\poly(\log n)$ and ideally a constant.

A dynamic algorithm for approximate VC (respectively, matching) should maintain a data structure
that answers queries of whether a vertex is in the VC (resp., an edge is matched) or not in \emph{constant} time.
Constant query time is considered a standard requirement in this line of research, and the goal is to optimize the update time of the algorithm under this requirement.
Almost all related works follow another requirement, of bounding the number of \emph{changes} to the maintained structure per step, either in the amortized or in the worst-case sense. The update time of the algorithm, which is the time it needs to update the data structure,
may be significantly lower than, and is    bounded by, the number of changes to the maintained structure;
for a motivation of this requirement, refer to \cite{BLSZ14,ADKKP16}.  
It is easy to see  that maintaining an exact minimum VC or a maximum matching requires $\Omega(n)$ changes per update even in the amortized sense, and even
for a simple path that changes dynamically in a straightforward way.  

Except for general graphs, these  problems have been studied mostly in bounded arboricity graphs \cite{NS13,KKPS14,HTZ14,BS15,BS16,PS16}.
It was shown in \cite{NS13} that a maximal matching can be maintained with amortized time $O(\log n / \log\log n)$ in constant arboricity graphs,
and this bound was improved in \cite{HTZ14}  to $O(\sqrt{\log n})$.
\cite{KKPS14} achieved a worst-case update time of $O(\log n)$.
The algorithms of \cite{NS13,HTZ14,KKPS14}   extend to graphs with arboricity bounded by $\alpha$, with the update time depending on $\alpha$.
A randomized algorithm for maintaining a maximal matching in \emph{general graphs} with constant amortized update time was given in \cite{Sol16}.
A maximal matching provides a 2-approximation for both the  maximum matching and the minimum VC.

What about better-than-2 approximations?
Improving upon \cite{BS15,BS16} and providing essentially the best result one can hope for in graphs of arboricity  $\alpha$,
 \cite{PS16} showed that a $(1+\eps)$-maximum matching can be maintained with a worst-case update time of $O(\alpha)$.
The $O(\alpha)$ bound in \cite{PS16} also bounds the number of changes to the matching,
and as mentioned it is impossible to maintain an exact matching with $o(n)$ matching changes even for a dynamic path.
In addition, \cite{PS16} showed that a $(2+\eps)$-VC can be maintained with a worst-case update time of $O(\alpha)$.
This improves the update time of the 2-VC algorithms of \cite{NS13,HTZ14,KKPS14} in every aspect, at the expense of increasing the approximation guarantee from 2 to $(2+\eps)$.

Note that in general graphs, a better-than-2 approximation to the minimum VC cannot be maintained efficiently under the unique games conjecture \cite{KR08}.
Although this hardness result does not apply to bounded arboricity graphs, there is currently no dynamic algorithm for maintaining a better-than-2 approximate VC with
update time $o(n)$ even in the amortized sense, and even in dynamic forests!\footnote{The only exception is an algorithm for maintaining a maximum matching
in dynamic forests with a worst-case update time of $O(\log n)$ \cite{GS09}.
As a result one can maintain the \emph{size} of the minimum VC (by Konig's theorem) in logarithmic update time.
On the negative side, one cannot efficiently maintain the VC itself or even a poor approximation of it  using \cite{GS09},
and more importantly, the result of \cite{GS09} requires a logarithmic query time, hence it does not follow the standard constant query time requirement.
(We believe that \cite{GS09} is the only paper in this line of research that does not follow this requirement.)}
In fact, the only known way to maintain a better-than-2 VC dynamically is to apply the fastest static algorithm from scratch following every update step.

The   local nature of our sparsification algorithms can be used to amplify the approximation guarantee of our VC sparsifier in time roughly linear in its size.
As a corollary, we provide a reduction from the dynamic to the static case,
showing that if a $t$-VC can be computed from scratch in time $T(n)$ in any (sub)family of graphs with arboricity bounded by $\alpha$, for any $t \ge 1$,
then a $(t+\eps)$-VC can be maintained with a worst-case update time of $\frac{T(n)}{O((n / \alpha) \cdot \eps^2)}$. This 
bound of $\frac{T(n)}{O((n / \alpha) \cdot \eps^2)}$ also bounds the amortized number of changes to the VC.
For planar graphs this yields an algorithm for maintaining an $(1+\eps)$-VC with a constant worst-case update time for any constant $\eps > 0$,
which is essentially the best one can hope for. For graphs of arboricity bounded by $\alpha$ we can maintain a VC of
approximation guarantee $\approx 2 - \frac{1}{\alpha}$ with a worst-case update time of $O(\sqrt{n} \cdot \alpha^2)$.

We can also amplify our matching and IS sparsifiers and obtain reductions from the dynamic to the static case.
These reductions are not useful to obtain new time bounds for the maximum matching and IS problems. In particular, for approximate matchings, the result of \cite{PS16} is already the best one can hope for;  
nevertheless, our reduction for the maximum matching problem can be used to obtain simpler and cleaner algorithms and arguments than those of \cite{PS16}, as discussed in Section \ref{s42}.
\subsection{Distributed networks} \label{s14}
Our sparsification algorithms can be implemented within a single communication round in distributed networks, where each processor sends and receives  a single $O(1)$-bit message along each of its adjacent edges.
Moreover, if $\Delta$ is the maximum degree of the sparsifier,   each processor may send messages along just $\Delta$ of its adjacent edges,
which ensures that the   \emph{load} on all the processors  will be low throughout the sparsifier's computation.
After the sparsifier has been computed, we can run on it the required distributed algorithm rather than on the original network,
which may significantly reduce the total runtime of the algorithm, the load on the processors, and in some settings it may also reduce the local memory usage at a processor.

Our distributed sparsification algorithms directly extend results from bounded degree graphs to bounded arboricity graphs or to graphs of bounded average degree,
for all the problems studied in this paper.
Since the performance of many distributed algorithm depend on the maximum degree of the underlying network
and as our sparsification algorithms are extremely simple, we anticipate that they will be used and implemented in practice.  

For the distributed approximate VC problem, \cite{BCS16} showed how to compute a $(2+\eps)$-VC in $O(\log \Delta / (\eps \log\log \Delta))$ rounds,
where $\Delta$ is the maximum degree in the graph.
We can plug our reduction to extend the result of \cite{BCS16} to graphs of arboricity bounded by $\alpha$, 
getting a $(2+\eps)$-VC in $O(\log (\alpha / \eps) / (\eps \log \log (\alpha / \eps)))$ rounds.

For the distributed approximate matching problem, a reduction from bounded arboricity graphs to bounded degree graphs was already given in \cite{CHS09}.
Nonetheless, our reduction has several advantages over that of \cite{CHS09} (see Section \ref{app}), one of which is that it is much simpler,
another is that our degree bound has better dependence on $\eps$.
In particular, \cite{EMR15} devised a distributed algorithm for computing a $(1+\eps)$-maximum matching in $\Delta^{O(1/\eps)} + O(\eps^{-2}) \cdot \log^* n$
rounds. Plugging our reduction (instead of that from \cite{CHS09}), we easily extend the result of \cite{EMR15} to graphs of arboricity bounded by $\alpha$
to get a $(1+\eps)$-maximum matching in $(\alpha / \eps)^{O(1/\eps)} + O(\eps^{-2}) \cdot \log^* n$ rounds.

A reduction from bounded arboricity graphs to bounded degree graphs was given in \cite{BEPS16} for the problems of maximal matching,
maximal IS, vertex coloring and ruling sets. The reduction of \cite{BEPS16} is based on different ideas than ours
(their algorithm is randomized, the number of rounds required by their algorithm is polylogarithmic in the maximum degree, etc),
and moreover, it appears that the reduction of \cite{BEPS16} cannot be efficiently applied to the problems studied in this paper.


\subsection{Subsequent work}
This paper appeared in the proceedings of ITCS 2018 \cite{Sol18}, and it triggered
further work in the area. We next survey the most relevant follow-up papers \cite{SS21,KS21,BDF+19,MS20}.

Kaplan and Solomon \cite{KS21} studied dynamic representations of distributed networks of bounded arboricity,
and they provided the first representation of distributed networks in which the local memory usage at all vertices is bounded by the arboricity rather than the maximum degree.
Among other results, they showed that the matching sparsifier presented in the current work 
can be maintained dynamically in a distributed network using low local memory usage; consequently, \cite{KS21} achieves efficient distributed algorithms for maintaining approximate matching and vertex cover 
in distributed networks of  bounded arboricity with low amortized update time and message complexities and with low local memory usage.

Behnezhad et al.\ \cite{BDF+19} studied the {\em stochastic matching problem}, where one is given a graph $G$ where each edge is realized independently with some constant probability $p$, and the goal is to find a constant degree subgraph of $G$ whose expected realized matching size is close to that of $G$.
This model of stochastic matching has received growing attention over the last few years due to its 
applications in kidney exchange and recommendation systems.
\cite{BDF+19} used the matching sparsifier presented in the current work for the stochastic setting.
More specifically, the analysis of the current work carries over to the stochastic setting without any change if $p = 1$, and \cite{BDF+19} showed that with few changes, the same analysis can be used to show that this matching sparsifier also works in the stochastic setting when $p < 1$.

Milenković and Solomon \cite{MS20} presented a matching sparsification algorithm in
graphs of bounded {\em neighborhood independence}.
The {\em neighborhood independence number} of a graph $G$, denoted by
$\beta = \beta(G)$, is the size of the largest independent set in the neighborhood of any vertex. 
Graphs with bounded neighborhood independence, even for constant $\beta$, form a wide family of possibly dense graphs, including line graphs, unit-disk graphs and graphs of bounded growth.
The matching sparsifier of \cite{MS20} is reminiscent of the matching sparsifier given in this work, with one key difference: it crucially relies on randomization. 
Set $\Delta = O(\frac{\beta}{\eps}\log\frac{1}{\eps})$ and let $G_\Delta$ be a random subgraph of $G$ that includes, for each vertex $v \in G$, $\Delta$ random edges incident on it;
\cite{MS20} showed that, with high probability, $G_\Delta$ is a {$(1+\eps)$-maximum matching sparsifier} for $G$.
Although the maximum degree of $G_\Delta$ can be significantly larger than $\Delta$, it is easy to see that the arboricity of $G_\Delta$ is at most $2\Delta$. Thus, one can apply the bounded degree $(1+\eps)$-maximum matching sparsifier of the current work on top of it to achieve a $((1+\eps)^2)$-maximum matching sparsification for the original graph $G$. This two-step sparsification approach, which can be implemented in two rounds of communication, was employed by \cite{MS20} to achieve a fast distributed algorithm for computing a $(1+\eps)$-maximum matching
in bounded neighborhood independence graphs.

The update time bounds of all our dynamic algorithms hold in the worst-case.
However, the number of {\em changes} per update step to the maintained structure (also known as the {\em recourse bound}) may be as large as $\Omega(n)$ in the worst-case. 
This is a rather common drawback of dynamic graph algorithms, where, for certain problems, the best worst-case recourse bounds are much larger than the corresponding update times.
In particular, this was the case for matchings, but the issue for matchings was remedied via a general black-box reduction by Solomon and Solomon \cite{SS21}, which shows that any dynamic algorithm for maintaining a $\gamma$-maximum matching with update time $T$, for any $\gamma \ge 1, T$ and $\eps > 0$, can be transformed into an algorithm for maintaining a $(\gamma(1 + \eps))$-maximum matching with update time $T + O(1/\eps)$ and a worst-case recourse bound of $O(1/\eps)$.

\subsection{Organization} \label{s15}
The definitions and notation used throughout are given in Section \ref{prel}.
Our matching, VC and IS sparsifiers are presented in Sections \ref{s31}, \ref{s32} and \ref{s33}, respectively.
Finally, in Section \ref{app} we provide some applications of our sparsifiers.

\section{Preliminaries} \label{prel}
Consider an unweighted graph $G = (V,E)$.
A matching $M$ for $G$ is said to be \emph{almost-maximal} w.r.t.\ some parameter $\eta> 0$, or  \emph{$\eta$-maximal}, if at most $\eta \cdot |M|$ edges can be added to it
while keeping it a valid matching for $G$. 
A $(1+\eps)$-maximal matching sparsifier for $G$ is a subgraph $H$ of $G$, such that any $\eta$-maximal matching for $H$ is an $(\eps+O(\eta))$-maximal matching for $G$;
in particular, a maximal matching for $H$ is   $\eps$-maximal   for $G$, and a $\eps$-maximal matching for $H$ is $O(\eps)$-maximal for $G$.

For a vertex $v$ in $G$, let $\Gamma_G(v)$ denote the set of neighbors (or \emph{neighborhood}) of $v$ in $G$.
For any vertex set $V' \subseteq V$  in $G$,   denote by $G[V']$ the subgraph induced by $V'$.

A graph $G=(V,E)$ has \emph{arboricity} $\alpha$ if $\alpha=\max_{U\subseteq V}\left\lceil\frac{|E(U)|}{|U|-1}\right\rceil$,
where $E(U)=\left\{(u,v)\in E\mid u,v\in U\right\}$. Alternatively, the arboricity of a graph is the minimum number of edge-disjoint forests into which it can be partitioned.
The family of bounded arboricity graphs
contains planar and bounded genus graphs, bounded tree-width graphs, and in general all graphs excluding fixed minors.

As mentioned in Section \ref{s11},  for any   $\Delta \ge 1$,
we write $V^\Delta_{high}$ and $V^\Delta_{low}$ to denote the sets of vertices of degree $\ge \Delta$ and $< \Delta$, respectively,
omitting $\Delta$ from the superscript when it is clear from the context. Define $G_{high} = G[V_{high}]$ and $G_{low} = G[V_{low}]$.





\section{Our Sparsifiers}

Note that a graph with arboricity $\alpha$ has an average degree at most $2\alpha$.
Throughout we use $\alpha$ as an arboricity parameter and $\beta$ as an average degree parameter.
The next observation will be useful.
\begin{observation} \label{basic}
Let $G = (V_1 \cup V_2, E)$ be a graph with average degree bounded by $\beta$, and suppose that each vertex of $V_1$ has degree at least $(c + 1) \beta$,
 for any $c$.
Then $|V_1| \le  |V_2|/c$.
\end{observation}
\begin{proof}
Observe that $2|E| \le \beta (|V_1| + |V_2|)$.
Since every vertex in $U$ has degree at least $(c+ 1)  \beta$, we have $2|E| \ge |V_1| \cdot (c + 1)  \beta$,
hence $|V_1| \cdot (c + 1) \beta \le \beta (|V_1| + |V_2|)$, and so $|V_1| \le |V_2|/c$.
\end{proof}
\subsection{The matching sparsifier}  \label{s31}
Let $G$ be a graph of arboricity bounded by $\alpha$, set $\Delta = 5(5/\eps + 1) 2\alpha$, and define the sets $V_{high}, V_{low}$ and the subgraphs $G_{low}, G_{high}$ accordingly.
We assume that $\eps \le 1$; the argument works  also for larger  $\eps$, by increasing $\Delta$ appropriately.
(We did not try to optimize the constants in the definition of $\Delta$.)

Recall our definition of the  matching sparsifier $G_\Delta$:
Mark up to $\Delta$ arbitrary adjacent edges on every vertex $v$, and add to $G_\Delta$ all edges that are marked  by both endpoints.
It is clear that the degree of $G_\Delta$ is at most $\Delta$.
To prove that $G_\Delta$ is a matching sparsifier, we use Hall's marriage theorem.
\begin{theorem} [Hall's marriage theorem \cite{Hall35}] \label{hall}
Let $G$ be a bipartite graph with sides $X$ and $Y$.
There is a matching that entirely covers $X$ if and only if for every subset $W$ of $X$, $|W| \leq |\Gamma_G(W)|$,
where $\Gamma_{G}(W) = \bigcup_{v \in W} \Gamma_{G}(v)$ is the \emph{neighborhood} of $W$.  
\end{theorem}

The following theorem shows that $G_\Delta$ is a $(1+\eps)$-maximum matching sparsifier.
\begin{theorem} \label{core}
Let $G$ be a graph of arboricity bounded by $\alpha$ and define $G_\Delta$ as above, for $\Delta = 5(5/\eps + 1) 2\alpha,\eps \le 1$.
Also, denote by $\cM^*$ and $\cM^*_{\Delta}$ the maximum matchings for $G$ and $G_{\Delta}$, respectively.
Then $|\cM^*| \le (1+\eps) \cdot |\cM^*_{\Delta}|$.
(In particular, any $t$-matching for $G_\Delta$ is a $t(1+\eps)$-matching for $G$.)
\end{theorem}
\begin{proof}
We shall construct a matching $\cM_{\Delta}$ for $G_{\Delta}$ satisfying $|\cM^*| \le (1+\eps) \cdot |\cM_{\Delta}|$.

Let $\cM^*_1$ be the subset of $\cM^*$
of all edges  that belong
to $G_{\Delta}$, and initialize $\cM_{\Delta}$ to $\cM^*_1$.
Define $\cM^*_2= \cM^* \setminus \cM^*_1$ as the complementary subset of $\cM^*$. We next show that
sufficiently many additional edges of $G_{\Delta}$ can be added to $\cM_{\Delta}$ while keeping it a valid matching.

Note that any edge with two endpoints  in $V_{low}$ must belong to $G_{\Delta}$.
Since the edges of $\cM^*_2$ do not belong to $G_{\Delta}$ by definition,
any edge of $\cM^*_2$ is adjacent on at least one vertex of $V_{high}$.

Let $V^{in} = V^{in}_{high}$ be the subset of $V_{high}$ of all vertices with at least $2\Delta/5$ neighbors in $V_{high}$, and let $V^{out} = V^{out}_{high} = V_{high} \setminus V^{in}$ be the complementary subset of $V_{high}$.
Observe that $G_{high}$ has arboricity at most $\alpha$, and thus average degree at most $2\alpha$.
Moreover, every vertex in $V^{in}$ has degree at least $2\Delta/5 = 2(5/\eps + 1) 2\alpha \ge (10/\eps + 1) 2\alpha$ in $G_{high}$. Observation \ref{basic} thus yields
\begin{equation} \label{in1}
|V^{in}| ~\le~  \eps/10 \cdot |V^{out}|.
\end{equation}
\begin{observation} \label{obsvout}
Each vertex in $V^{out}$ has more than $3\Delta/5$ neighbors in $V_{low}$ within $G_\Delta$.
\end{observation}
\begin{proof}
First, note that any vertex in $V_{low}$ marks all its $< \Delta$ adjacent edges.
Since each vertex in $V^{out}$ has $< 2\Delta/5$ neighbors in $V_{high}$ but a degree of $\ge \Delta$,
the remaining $> 3\Delta/5$ neighbors must be in $V_{low}$. Each vertex of $V^{out}$ marks $\Delta$ edges, and any of the $> 3\Delta/5$ edges adjacent to a vertex of $V_{low}$ is also marked by
that endpoint in $V_{low}$, and is thus added to $G_\Delta$.
\end{proof}

For a vertex $v$, we will refer to its neighbors within $G_\Delta$ as its \emph{$G_\Delta$-neighbors}.
Let $U = U^{out}$ be the set of vertices in $V^{out}$ that are free w.r.t.\ $\cM^*_1$.
By Observation \ref{obsvout}, each vertex of  $U$ has more than $3\Delta/5$ $G_\Delta$-neighbors in $V_{low}$.
Denote by $\Gamma = \Gamma(U)$ the set of all $G_\Delta$-neighbors of vertices from $U$ in $V_{low}$.
We next partition $\Gamma$ into two sets, the sets ${\Gamma}_{matched}$ and ${\Gamma}_{free}$ of vertices that are matched and free  w.r.t.\ $\cM^*_1$, respectively.
A vertex $u \in U$ is called \emph{risky} if at least $2\Delta/5$ of its $G_\Delta$-neighbors
in $\Gamma$ are in ${\Gamma}_{matched}$, otherwise it is \emph{safe}, and then more than $\Delta/5$ of its $G_\Delta$-neighbors are in  ${\Gamma}_{free}$.
Let $U_{risky}$ and $U_{safe}$ be the sets of all risky and safe vertices of $U$, respectively.
\begin{claim} \label{risky}
$|U_{risky}| \le  \eps/5 \cdot |\cM^*_1|$.
\end{claim}
\begin{proof}
For each vertex $u \in U_{risky}$, let $\Gamma_{matched}(u)$ be the set of its at least $2\Delta/5$ $G_\Delta$-neighbors in ${\Gamma}_{matched}$.
Let  $\Gamma_{risky} = \bigcup_{u \in U_{risky}} \Gamma_{matched}(u)$ denote the union of the sets $\Gamma_{matched}(u)$ over all vertices $u \in U_{risky}$.
Since all vertices in $\Gamma_{risky}$ are matched w.r.t.\ the matching $\cM^*_1$,
$|\Gamma_{risky}| \le 2|\cM^*_1|$.
Observe that the subgraph $G_{risky}$ of $G_\Delta$ induced by $U_{risky} \cup \Gamma_{risky}$ has arboricity at most $\alpha$,
and thus average degree at most $2\alpha$.
Moreover, each vertex in $U_{risky}$ has at least $2\Delta/5$ neighbors in $\Gamma_{risky}$, and thus its degree in $G_{risky}$ is at least $2\Delta/5 = 2(5/\eps + 1) 2\alpha
\ge (10/\eps + 1) 2\alpha$.
Observation \ref{basic} thus yields
$|U_{risky}| ~\le~ \eps/10 \cdot |\Gamma_{risky}| \le \eps/5 \cdot |\cM^*_1|$.
\end{proof}
Note that any edge in $G_\Delta$ between a vertex in $U_{safe}$ and a vertex in ${\Gamma}_{free}$
is vertex-disjoint to all edges of $\cM^*_1$.
Consequently, the following claim implies that at least $|U_{safe}|$ edges of $G_\Delta$ can be added to the matching $\cM_\Delta$ while preserving its validity.  
\begin{claim} \label{safe}
There exists a  matching $\cM_{safe}$ in $G_\Delta$ that entirely covers $U_{safe}$ by edges between $U_{safe}$ and ${\Gamma}_{free}$.
In particular, $|\cM_{safe}| = |U_{safe}|$.
\end{claim}
\begin{proof}
For each vertex $u \in U_{safe}$, let $\Gamma_{free}(u)$ be the set of its at least $\Delta/5$ $G_\Delta$-neighbors in ${\Gamma}_{free}$.
Let $\Gamma_{safe} = \bigcup_{u \in U_{safe}} \Gamma_{free}(u)$ denote the union of the sets $\Gamma_{free}(u)$ over all vertices $u \in U_{safe}$.
Consider the induced bipartite subgraph $G_{safe}$ of $G_\Delta$ with sides $U_{safe}$ and $\Gamma_{safe}$.
We argue that for any subset $W \subseteq U_{safe}$, its neighborhood in $G_{safe}$, namely,
$$\Gamma_{G_{safe}}(W) ~=~ \bigcup_{v \in W} \Gamma_{G_{safe}}(v) ~=~ \bigcup_{v \in W} \Gamma_{free}(v),$$  is of larger size.
Obsserve that the subgraph $G^W_{safe}$  of $G_{safe}$ induced by the vertex set $W \cup \Gamma_{G_{safe}}(W)$ has arboricity at most $\alpha$, and thus average degree at most $2\alpha$.
Moreover, each vertex of $W$ has at least $\Delta/5$ neighbors in $\Gamma_{G_{safe}}(W)$, i.e., its degree in $G^W_{safe}$ is at least $\Delta/5 =(5/\eps + 1) 2\alpha$.
Observation \ref{basic} thus yields $|W| \le \eps / 5 \cdot |\Gamma_{G_{safe}}(W)| < |\Gamma_{G_{safe}}(W)|$.
By Hall's marriage theorem (Theorem \ref{hall}), there exists a matching $\cM_{safe}$ in $G_{safe}$ that entirely covers $U_{safe}$.
Claim \ref{safe} follows.
\end{proof}
We add all edges in the matching $\cM_{safe}$ guaranteed by Claim \ref{safe} to  $\cM_\Delta$, so that $\cM_\Delta = \cM^*_1 \cup \cM_{safe}$.
Since $\cM^*$ is a maximum matching for $G$ that is a disjoint union of $\cM^*_1$ and $\cM^*_2$ and as $\cM_\Delta$ is a disjoint union of $\cM^*_1$ and $\cM_{safe}$, it follows that $|\cM^*_2| \ge |\cM_{safe}|$.
Although the vertices of $V^{in}$ may be matched w.r.t.\ $\cM^*_2$,   we have $|V^{in}| \le \eps/10 \cdot |V^{out}|$ by Equation (\ref{in1}).
By definition, all vertices of $V^{out} \setminus U$ are matched w.r.t. $\cM^*_1$, thus $|V^{out} \setminus U| \le 2|\cM^*_1||$.
Hence $|V^{out}| \le 2|\cM^*_1| + |U|$, and so
\begin{equation} \label{in}
|V^{in}| ~\le~ \eps/10 \cdot |V^{out}| ~\le~ \eps /10 \cdot (2|\cM^*_1| +  |U_{risky}| + |U_{safe}|).
\end{equation}
Since any edge of $\cM^*_2$ is adjacent on at least one vertex of $V_{high}$ and as all vertices of $V^{out} \setminus U$ are matched w.r.t. $\cM^*_1$,
 $|\cM^*_2| \le  |V^{in}| + |U_{risky}| + |U_{safe}|$.
Combined with Equation (\ref{in}) and Claim \ref{risky},
\begin{eqnarray*}
|\cM^*_2| &\le&  |V^{in}| + |U_{risky}| + |U_{safe}| ~\le~ \eps /10 \cdot (2|\cM^*_1| +  |U_{risky}| + |U_{safe}|) + |U_{risky}| + |U_{safe}|
\\ &\le& \eps/10 \cdot (2|\cM^*_1| +  \eps/5 \cdot |\cM^*_1| + |M^*_2|) + \eps/5 \cdot |\cM^*_1| + |M_{safe}|
\\ &=& (\eps/5 + \eps^2 / 50 + \eps/5) \cdot |\cM^*_1| + \eps/10 \cdot |M^*_2| + |M_{safe}|
~\le~ \eps/2 \cdot |M^*_1| + \eps/10 \cdot |M^*_2| + |M_{safe}|,
\end{eqnarray*}
yielding
\begin{eqnarray*}
|\cM_\Delta| &=& |\cM^*_1| + |\cM_{safe}|
~\ge~ |\cM^*_1| + |\cM^*_2| - \eps/2 \cdot |\cM^*_1| - \eps/10 \cdot |\cM^*_2|
\\ &\ge&  (1-\eps/2) \cdot (|\cM^*_1| + |\cM^*_2|) ~=~ (1-\eps/2) \cdot |\cM^*|.
\end{eqnarray*}
To complete the proof of Lemma \ref{core}, observe that
\begin{eqnarray*}
|\cM^*| &\le& \frac{1}{1-\eps/2} \cdot |\cM_\Delta| ~\le~ (1+\eps) \cdot |\cM_\Delta| ~\le~ (1+\eps) \cdot |\cM^*_\Delta|. 
\end{eqnarray*}
\end{proof}
The following theorem shows that $G_\Delta$ is a $(1+\eps)$-maximal matching sparsifier.
\begin{theorem} \label{amaximalpf}
Let $G$ be a graph of arboricity bounded by $\alpha$ and $G_\Delta$ defined as above, for $\Delta = 5(5/\eps + 1) 2\alpha,\eps \le 1$.
Any $\eta$-maximal matching for $G_{\Delta}$ is an $(\eps + 3\eta)$-maximal matching for $G$, for any $\eta > 0$.
\end{theorem}
\begin{proof}
The proof of this theorem is similar to that of Theorem \ref{core}, thus we aim for conciseness.

Consider an arbitrary $\eta$-maximal matching $\cM_{\Delta} = \cM^\eta_{\Delta}$ for $G_{\Delta}$,
and let $\cM'_\Delta$ be any   matching for $G$ obtained by adding edges to $\cM_\Delta$. We will show that $|\cM'_\Delta \setminus \cM_\Delta| \le (\eps + 2\eta) \cdot |\cM_\Delta|$.

Since $\cM_\Delta$ is $\eta$-maximal w.r.t.\ $G_\Delta$,
at most $\eta \cdot |\cM_\Delta|$ edges of $\cM'_\Delta \setminus \cM_\Delta$ may belong to $G_\Delta$; in what follows we refer to those edges as \emph{special edges} and to
the remaining edges of $\cM'_\Delta \setminus \cM_\Delta$ as ordinary edges. Denote the set of ordinary edges in $\cM'_\Delta \setminus \cM_\Delta$ by $O$.
Since any edge with two endpoints in $V_{low}$ belongs to $G_\Delta$, it follows that any edge of $O$ has at least one endpoint in $V_{high}$.
Denote the set of vertices in $V_{high}$ that are free w.r.t.\ $\cM_\Delta$ by $F_{high}$, and note that $|O| \le |F_{high}|$.

Defining the sets $V^{in}$ and $V^{out}$ as in the proof of Theorem \ref{core}, Equation (\ref{in1}) yields $|V^{in}| \le  \eps/10 \cdot |V^{out}|$.
Notice that $|V^{out}| \le 2|\cM_\Delta| + |F_{high} \cap V^{out}|$,
hence
\begin{equation} \label{highin}
|F_{high} \cap V^{in}| ~\le~ |V^{in}| ~\le~  \eps/10 \cdot |V^{out}| ~\le~ \eps/10 \cdot (2|\cM_\Delta| + |F_{high} \cap V^{out}|).
\end{equation}
Observation \ref{obsvout} from the proof of Theorem \ref{core} remains valid, and it implies that each vertex in $F_{high} \cap V^{out}$ has more than $3\Delta/5$ $G_\Delta$-neighbors in $V_{low}$.
Denote by $\Gamma = \Gamma(F_{high} \cap V^{out})$ the set of all $G_\Delta$-neighbors of vertices from $F_{high} \cap V^{out}$ in $V_{low}$.
We next partition $\Gamma$ into two sets, the sets ${\Gamma}_{matched}$ and ${\Gamma}_{free}$ of vertices in $\Gamma$ that are matched and free  w.r.t.\ $\cM_\Delta$, respectively.
A vertex $f \in F_{high} \cap V^{out}$ is called \emph{risky} if at least $2\Delta/5$ of its $G_\Delta$-neighbors
 are in ${\Gamma}_{matched}$, otherwise it is \emph{safe}, and then more than $\Delta/5$ of its $G_\Delta$-neighbors are in  ${\Gamma}_{free}$.
Let $F_{risky}$ and $F_{safe}$ be the sets of all risky and safe vertices of $F_{high} \cap V^{out}$, respectively.
Following similar lines as those in the proof of Claim \ref{risky},
we get $|F_{risky}| \le  \eps/5 \cdot |\cM_\Delta|$.
Also, following similar lines as those in the proof of Claim \ref{safe}, we establish the existence of a matching $\cM_{safe}$ in $G_\Delta$ that entirely covers $F_{safe}$ by edges between $F_{safe}$ and ${\Gamma}_{free}$.
Since all edges of $\cM_{safe}$ belong to $G_\Delta$ and can be added to $\cM_\Delta$ while keeping it a valid matching for $G_\Delta$,
the fact that $\cM_\Delta$ is $\eta$-maximal w.r.t.\ $G_\Delta$ yields
$|F_{safe}| \le \eta \cdot |\cM_\Delta|$.
It follows that $|F_{high} \cap V^{out}| = |F_{risky}| + |F_{safe}| \le \eps/5 \cdot |\cM_\Delta| + \eta \cdot |\cM_\Delta|$.
Recall that $|O| \le |F_{high}|$. Combined with Equation (\ref{highin}), we conclude that
 \begin{eqnarray*}
|O| &\le& |F_{high}| ~=~ |F_{high} \cap V^{in}| + |F_{high} \cap V^{out}|
\\&\le& \eps/10 \cdot (2|\cM_\Delta| + |F_{high} \cap V^{out}|) + \eps/5 \cdot |\cM_\Delta| + \eta \cdot |\cM_\Delta|
\\&\le& \eps/10 \cdot (2|\cM_\Delta| + \eps/5 \cdot |\cM_\Delta| + \eta \cdot |\cM_\Delta|) + \eps/5 \cdot |\cM_\Delta| + \eta \cdot |\cM_\Delta|
\\&\le& (\eps/5  + \eps^2 / 50 +  \eps \cdot \eta  / 10 + \eps/5 +  \eta) \cdot |\cM_\Delta|
~<~ (\eps + 2\eta) \cdot |\cM_\Delta|,
\end{eqnarray*}
i.e.,  the number $|O|$ of ordinary edges in $\cM'_\Delta \setminus \cM_\Delta$ is at most $(\eps + 2\eta) \cdot |\cM_\Delta|$.
Recall that the number of special edges in $\cM'_\Delta \setminus \cM_\Delta$ is at most $\eta \cdot |\cM_\Delta|$,
thus $|\cM'_\Delta \setminus \cM_\Delta| \le (\eps + 3\eta) \cdot |\cM_\Delta|$.  
\end{proof}

\subsection{The VC sparsifier} \label{s32}
Let $G$ be a graph of arboricity bounded by $\alpha$, set $\Delta = (1/\eps + 1) \cdot 2\alpha$, and define the sets $V_{high}, V_{low}$ and the subgraph $G_{low}$ accordingly.
To prove that the pair $(G_{low},V_{high})$ is a VC sparsifier, we use the next lemma.
\begin{lemma} \label{eqlarge}
Let $VC$ be an arbitrary VC for $G$,
and let $U = U_{high}$ be the set of vertices in $V_{high}$ that are not in $VC$. Then $|U| \le \eps \cdot |VC|$.
\end{lemma}
\begin{proof}
Denote by $\Gamma = \Gamma_{high}$ the set of neighbors in $G \setminus U$ of all vertices of $U$, and note that $\Gamma \subseteq VC$,
as otherwise $VC$ cannot be a VC for $G$.
Observe that the subgraph $G' = (U \cup \Gamma,  E') = G[U \cup \Gamma]$ induced by $U \cup \Gamma$ has arboricity at most $\alpha$, and thus average degree at most $2\alpha$.
Moreover, every vertex in $U$ has degree at least $\Delta$.
Observation \ref{basic} thus yields
$|U| ~\le~ \eps \cdot |\Gamma| ~\le~ \eps \cdot |VC|.$
\end{proof}

The following theorem shows that $(G_{low},V_{high})$ is a $(1+\eps)$-VC sparsifier.
\begin{theorem} \label{vcspa}
Let $G$ be a graph of arboricity bounded by $\alpha$.
Also, let $VC^*$ be a minimum VC for $G$, let $VC^t_{low}$ be a $t$-VC for $G_{low}$, for any $t \ge 1$, and define $\widetilde{VC} = VC^t_{low} \cup V_{high}$.
Then $\widetilde{VC}$ is a $(t+\eps)$-VC for $G$.  
\end{theorem}
\begin{proof}
Since $VC^*$ is a VC for $G$, $VC^* \cap V_{low}$ must be a VC for $G_{low}$.
Hence
$|{VC}^t_{low}| ~\le~  t \cdot |VC^* \cap V_{low}|$.
Denoting by $U = U_{high}$  the set of vertices in $V_{high}$ that are not in $VC^*$,
we have $|V_{high}| = |VC^* \cap V_{high}| + |U|$. Also, Lemma \ref{eqlarge}  yields $|U| \le \eps \cdot |VC^*|$.
Since $|{VC}^t_{low}| ~\le~  t \cdot |VC^* \cap V_{low}|$, it follows that
\begin{eqnarray*}
|\widetilde{VC}| &=& |{VC}^t_{low}| + |V_{high}| ~=~ |{VC}^t_{low}| + |VC^* \cap V_{high}| + |U| \\ &\le&
t \cdot |VC^* \cap V_{low}| + |VC^* \cap V_{high}| + \eps \cdot |VC^*|
~\le~ (t+\eps) \cdot |VC^*|. 
\end{eqnarray*}
\end{proof}


\subsection{The IS sparsifier} \label{s33}
Let $G$ be a graph of average degree bounded by $\beta$ (the arboricity may be much larger than $\beta$).
Set $\Delta = ((\beta+1)/\eps + 1) \cdot \beta$, and define the sets $V_{high}, V_{low}$ and the subgraph $G_{low}$ accordingly.
We assume that $\beta \ge 1$, as we may ignore isolated vertices (adding all of them to the independent set).
To show that $G_{low}$ is an IS sparsifier, we make the following observation.
\begin{observation} \label{basic2}
Let $V_1$ be any set of vertices in an arbitrary graph $G = (V,E)$, and let $V_2$ be the complementary subset of $V$.
Let $IS^*$, $IS^*_1$ and $IS^*_2$ be maximum independent sets for the graph $G$ and its subgraphs $G[V_1]$ and $G[V_2]$, respectively.
Then $|IS^*_1| + |IS^*_2| \ge |IS^*|$.
\end{observation}
\begin{proof}
Both $IS^* \cap V_1$ and $IS^* \cap V_2$ are ISs, hence  $|IS^* \cap V_1| \le |IS^*_1|,  |IS^* \cap V_2| \le |IS^*_2|$.
\end{proof}

The following theorem shows that $G_{low}$ is an $(1+\eps)$-IS sparsifier
\begin{theorem}
Let $G$ be a graph of average degree bounded by $\beta$.
Also, let $IS^*$ (respectively, $IS^*_{low}$) be a maximum IS for $G$ (resp., $G_{low}$), let $IS^t_{low}$ be a $t$-IS for $G_{low}$, for any $t \ge 1$.  
Then $IS^t_{low}$ is a $t(1+\eps)$-IS for $G$, for any $\eps < \beta$.  
\end{theorem}
\begin{proof}
Since $IS^t_{low}$ is an IS for $G_{low}$, it must also be an IS for $G$.

Since the average degree in $G = (V_{low} \cup V_{high},E)$ is bounded by $\beta$
and as every vertex in $V_{high}$ has degree at least $\Delta = ((\beta+1)/\eps + 1) \cdot \beta$, Observation \ref{basic} implies that $|V_{high}| \le  \eps \cdot (|V_{low}| /(\beta+1))$.

Since $\eps < \beta$, we have $\Delta = ((\beta+1)/\eps + 1) \cdot \beta > \beta$, hence there is at least one vertex of degree less than $\Delta$,
i.e., $V_{low}$ is non-empty.
Denote the average degree in $G_{low}$ by $\beta_{low}$.
Since every vertex in $V_{high}$ has degree at least $\Delta = ((\beta+1)/\eps + 1) \cdot \beta$,  we have
$|V_{high}|  \cdot ((\beta+1)/\eps + 1) \cdot \beta + |V_{low}| \cdot \beta_{low} \le 2|E| \le \beta (|V_{low}| + |V_{high}|)$.
It follows that 
$|V_{high}| \cdot ((\beta+1)/\eps) ~\le~ |V_{low}| (1 - \frac{\beta_{low}}{\beta})$, which, together with
the fact that $V_{low}$ is non-empty, yields $\beta_{low} \le \beta$.
By Turan's theorem, $|IS^*_{low}| \ge V_{low} / (\beta_{low} + 1)$, hence
$$|V_{high}| ~\le~  \eps \cdot (|V_{low}| /(\beta+1)) ~\le~ \eps \cdot (|V_{low}| /(\beta_{low} +1)) ~\le~ \eps \cdot |IS^*_{low}|.$$
By Observation \ref{basic2}, $|IS^*| \le |V_{high}| + |IS^*_{low}|$,  
so $|IS^*| ~\le~ (1+\eps) \cdot |IS^*_{low}| ~\le~ t(1+\eps) \cdot |IS^t_{low}|$.
\end{proof}

\section{Applications} \label{app}
\subsection{Local computation algorithms} \label{s41}
Each vertex $v$ is represented as a unique ID from $\{1,\ldots,n\}$,
and the   graph $G = (V,E)$ is given through an adjacency list oracle $\cO_G$ that answers neighbor queries:
given a vertex $v \in V$ and an index $i$, the $i$th neighbor of $v$ is returned if $v$'s degree is $\ge i$; otherwise a null sign is returned.
Consider first our matching sparsifier $G_\Delta$, obtained by
marking up to $\Delta$ arbitrary adjacent edges on every vertex $v$, and adding to $G_\Delta$ all edges that are marked  by both endpoints.
Recall that $\Delta = O(\alpha / \eps)$, where $\alpha$ is a bound on the arboricity of $G$.
Our goal is to simulate the execution of any local computation algorithm for approximate matching entirely within our bounded degree sparsifier,
and in this way to reduce the problem from bounded arboricity graphs to bounded degree graphs.
To this end we simply need to be \emph{consistent} about the adjacent edges of a vertex $v$ that we mark,
e.g., for every vertex $v$, mark its first $\Delta$ neighbors on its adjacency list.
Whenever a vertex is queried/probed, there is no need to probe any other neighbor of this vertex besides the first $\Delta$ on its adjacency list,
since none of the edges that lead to the other neighbors is in the sparsifier.
To determine which among these neighbors is also its neighbor in the sparsifier, we perform a symmetric probe for each of the (at most) $\Delta$ neighbors.
In this way any probing procedure of the original graph (which may contain high degree vertices) is restricted to the matching sparsifier,
at the cost of increasing the time complexity by at most a factor of $\Delta$;
this loss is considered negligible, since all the time and space complexities of algorithms in this area are anyway at least polynomial in $\Delta$.
Moreover, by Theorems \ref{core} and \ref{amaximalpf},
any $(1+\eps)$-maximum (respectively, $\eps$-maximal) matching computed for $G_\Delta$ provides a $(1+O(\eps))$-maximum
(resp., $O(\eps)$-maximal) matching for the original graph, thus there is only a negligible loss in the approximation guarantee.
We remark that the local computation algorithm of \cite{Fischer17} is actually an almost-maximal matching algorithm,
and the $(2+\eps)$-approximation guarantee holds as any $\eps$-maximal matching is also a $(2+O(\eps))$-maximum matching.
In this way we reduce
the problems of approximate-maximum matching and almost-maximal matching from graphs of arboricity bounded by $\alpha$ to graphs of degree bounded by
$O(\alpha/\eps)$.  
\begin{corollary}
For any graph $G$ with arboricity bounded by $\alpha$ and for any constant $\eps > 0$:
\begin{itemize}
\item There is a deterministic local computation algorithm for $(1+\eps)$-maximum matching with time complexity $O(\log^* n) \cdot \exp(\alpha)$ and zero space complexity.
(An extension of \cite{EMR14}.)
\item There is a randomized local computation algorithms for $(1+\eps)$-maximum matching with time and space complexities of $\poly(\log n,\alpha)$. (An extension of \cite{LRY17}.)
\item There is a deterministic local computation algorithm for $(2+\eps)$-maximum matching with time complexity $O(\log^* n) \cdot 2^{O(\alpha^2)}$.
(An extension of \cite{Fischer17}.)
\end{itemize}
\end{corollary}
For our VC and IS sparsifiers, things are even simpler.
The IS sparsifier is simply $G_{low}$, i.e., the subgraph induced on the low degree vertices, hence simulating the execution of a local computation algorithm entirely within
the sparsifier can be naturally done without having to probe more than $\Delta$ neighbors of any vertex.
As for the VC sparsifier, we also need to reason about the validating set $V^{high}$, but for any vertex we can determine if it belongs to $V^{high}$ or not
by making one query to the oracle $\cO_G$.
\subsection{Dynamic centralized algorithms} \label{s42}
In this section we employ our sparsifiers to get efficient dynamic algorithms.
The starting point is a ``lazy scheme'' due to \cite{GP13} for maintaining approximate maximum matching for general graphs,
which was refined for bounded arboricity graphs in \cite{PS16}.
This scheme exploits a basic \emph{stability} property of matchings: The size of the maximum matching changes by at most 1 following each update.
Thus if we have a \emph{large} matching  of size close to the maximum,
 it will remain close to it throughout a long update sequence, or formally:
 \begin{lemma} [{\bf Lemma 3.1 in \cite{GP13}}] \label{lazylemma}
Let $\eps,\eps' \le 1/2$. Suppose that $\cM_i$ is a $(1+\eps)$-MCM for $G_i$.
For $j = i,i+1,\ldots, i+\lfloor \eps'\cdot |\cM_i| \rfloor$, let $\cM^{(j)}_i$ denote the matching $\cM_i$ after removing from it all edges that got deleted during the updates $i+1,\ldots,j$.
Then $\cM^{(j)}_i$ is a $(1+2\eps+2\eps')$-MCM for the graph $G_j$.
\end{lemma}

%
Hence, we can compute a $(1+\eps/4)$-maximum matching $\cM_i$ at a certain update $i$, and use the same  matching $\cM^{(j)}_i$ throughout all updates $j = i,i+1,\ldots, i' =  i+ \lfloor \eps/4\cdot |\cM_i| \rfloor$. (By Lemma \ref{lazylemma}, $\cM^{(j)}_i$ is a $(1+\eps)$-maximum matching for all graphs $G_j$.)
Next compute a fresh $(1+\eps/4)$-maximum matching $\cM_{i'}$ following update $i'$ and use it throughout all updates $i',i'+1,\ldots,i'+ \lfloor \eps/4 \cdot |\cM_{i'}|\rfloor$, and repeat.

In this way the static time complexity of computing a $(1+\eps)$-maximum matching $\cM$
is amortized over $1 + \lfloor \eps/4\cdot |\cM| \rfloor  = \Omega(\eps \cdot |\cM|)$ updates.
The key insight behind the schemes of \cite{GP13,PS16} is not to compute the approximate matching on the entire graph, but rather on a matching sparsifier,
which is derived from an $O(1)$-VC that is maintained dynamically \emph{by other means}. \cite{PS16} showed that an $\eps$-maximal matching, and thus a $(2+\eps)$-VC,
denoted by $\widetilde{VC}$, can be maintained with $O(\alpha/\eps)$ update time, and the argument used by \cite{PS16} was quite tricky; we will get back to this point soon.
Specifically, the sparsifier $\tilde G$ of \cite{PS16} contains (1) all edges in the subgraph $G[\widetilde{VC}]$ induced by $\widetilde{VC}$,
and (2) for each vertex $v \in \widetilde{VC}$, (up to) $O(\alpha/\eps)$ edges connecting it with arbitrary neighbors outside $\widetilde{VC}$.
Although the resulting sparsifier may have large degree, it is easy to verify that it contains
$O(|\cM| \cdot \alpha/\eps)$ edges, and it can be computed in time linear in its size.
Since a $(1+\eps)$-maximum matching can be computed for the sparsifier $\tilde G$ in time $O(|\tilde G| / \eps) = O(|\cM| \cdot \alpha/\eps^2)$ \cite{HK73,MV80,Vaz12},
the resulting amortized update time is $O(\alpha \cdot \eps^{-3})$. Also, one can easily translate the amortized bound into a worst-case bound, as shown in \cite{GP13}.

The aforementioned stability property also applies to the minimum VC and maximum IS problems.
We next consider the minimum VC problem; see Section \ref{ISS} for a discussion on the maximum IS problem.
Thus the size of the minimum VC changes by at most 1 following each update.
In extending the lazy scheme \cite{GP13,PS16} to the minimum VC problem, the challenge is to efficiently compute a high quality VC sparsifier.
In particular, the $(1+\eps)$-matching sparsifier $\tilde G$ of \cite{PS16} cannot be used as such a VC sparsifier,
since a VC for it (regardless of its approximation guarantee) may not provide a \emph{valid} VC for the graph $G$.

Fix an arbitrary parameter $t \ge 1$.
Consider an arbitrary (sub)family of $n$-vertex graphs $\mathcal{G}$ with arboricity bounded by $\alpha$ that is closed under edge removals (such as the family of planar graphs),
and suppose that we can compute a $t$-VC from scratch in time $T(n)$ for any graph in this family.  
(More generally, we assume that for any $j$-vertex subgraph $H$ of any $G \in \mathcal{G}$, for any $1 \le j \le n$, we can compute in time $T(j)$ a $t$-VC for $H$.)
Next, we adapt the lazy scheme \cite{GP13,PS16} to maintain a $(t+2\eps)$-VC in dynamically changing graphs of $\mathcal{G}$.

Lemma \ref{lazylemma} easily extends to the minimum VC problem and to any approximation   $t \ge 1$, i.e., any $(t+\eps)$-VC
continues to provide a $(t+2\eps)$-VC throughout a long update sequence.
Once the approximation guarantee of that VC, denoted by $VC_{old}$, becomes too poor (i.e., reaches $t+2\eps$), 
we shall \emph{amplify} it (i.e., reduce it back to $t+\eps$) by computing a $(t+\eps)$-VC within time close to linear in $|VC_{low}|$.  
Having done that, we can then re-use the new amplified VC, denoted by ${VC}_{new}$, throughout the subsequent $\eps \cdot |{VC}_{new}|$ update steps, and repeat.

The computation of the new amplified cover $VC_{new}$ employs the old cover $VC_{old}$ and our VC sparsifier from Section \ref{s32} as follows.
Recall that the vertex set $V_{high}$ is the validating set of our VC sparsifier.
It is straightforward to maintain all vertices in $V_{high}$ dynamically with $O(\alpha / \eps)$ worst-case update time;
we initialize ${VC}_{new}$ as $V_{high}$.
We next need to compute the subgraph $G_{low}$ induced on the vertices $V_{low}$ of degree $< \Delta = O(\alpha / \eps)$.
This can be done by simply adding, for each vertex of ${VC}_{old}$ of degree $< \Delta$, all its adjacent edges to vertices of degree $< \Delta$.
Although  some vertices of $V_{low}$ may not belong to $VC_{old}$, we must have added all edges of $G_{low}$, as $VC_{old}$ is a VC for $G$.
Clearly, the number $|V_{low}|$ of vertices in $G_{low}$, as well as the time needed to compute it, are bounded by $O(|{VC}_{old}| \cdot \alpha / \eps)$. 
We proceed by computing a $t$-VC for $G_{low}$, denoted by $VC^t_{low}$; the runtime of this static computation is $T(|V_{low}|)$.
Finally, we add all vertices of $VC^t_{low}$ to ${VC}_{new}$, thus we have ${VC}_{new} = VC^t_{low} \cup V_{high}$. 
By Theorem \ref{vcspa},  ${VC}_{new}$ is a $(t+\eps)$-VC for the entire graph $G$.

Observe that the overall runtime of computing ${VC}_{new}$
is bounded by $O(|{VC}_{old}| \cdot \alpha / \eps) + T(|V_{low}|) \le T(O(|{VC}_{old}| \cdot \alpha / \eps))$.
Since this runtime is amortized over $\Theta(\eps \cdot |{VC}_{old}|)$ update steps, the amortized update time is bounded by 
$T(O(|{VC}_{old}| \cdot \alpha / \eps))/ \Theta(\eps \cdot |{VC}_{old}|)$, which is no greater than
$\frac{T(n)}{O((n / \alpha) \cdot \eps^2)}$.
Note  that the amortized number of changes to the VC is also bounded by $\frac{T(n)}{O((n / \alpha) \cdot \eps^2)}$.
Moreover, one can easily translate the amortized update time into the same (up to a constant factor) worst-case update time, as shown in \cite{GP13}.

Summarizing, we have proved the following result.
\begin{theorem} \label{thmm}
Fix arbitrary $t \ge 1, \eps > 0$.
If a $t$-VC can be computed in time $T(n)$ in an arbitrary (sub)family of $n$-vertex graphs with arboricity bounded by $\alpha$ that is closed under edge removals,
then a $(t+\eps)$-VC can be maintained with a worst-case update time of $\frac{T(n)}{O((n / \alpha) \cdot \eps^2)}$.  
Moreover, the amortized number of changes to the VC is also bounded by $\frac{T(n)}{O((n / \alpha) \cdot \eps^2)}$.
\end{theorem}
{\bf Remarks.}
\begin{enumerate}
\item For planar graphs, one can compute a $(1+\eps)$-VC in time $O(n)$, for any constant $\eps > 0$ \cite{Baker94}.
By Theorem \ref{thmm}, we can   maintain a $(1+\eps)$-VC with a constant  worst-case update time for any constant $\eps > 0$.
\item  For the family of graphs with arboricity bounded by $\alpha$, one can compute a $(2 - \frac{2}{\beta+1})$-VC in time $O(n^{3/2} \cdot \alpha)$, where $\beta$ is the average degree in some (carefully computed) subgraph, and is thus bounded by $2\alpha$ \cite{Hochbaum83}.
By Theorem \ref{thmm}, we can maintain a VC with an approximation of roughly $2 - \frac{1}{\alpha}$ and a worst-case update time of $O(\sqrt{n} \cdot \alpha^{2})$.
\end{enumerate}

Our matching sparsifier from Section \ref{s31} can be used to simplify the algorithms of \cite{PS16} and their analysis.
First, as mentioned, \cite{PS16} used a tricky argument to maintain a $(2+\eps)$-VC with $O(\alpha/\eps)$ update time.
An alternative simpler approach is to maintain a maximal matching on top of our bounded degree matching sparsifier, which can be done naively with update time linear in the
degree of the sparsifier, namely, $O(\alpha/\eps)$. By Theorem \ref{amaximalpf}, this yields an $\eps$-maximal matching, 
which is translated into an $(2+\eps)$-VC in the obvious way.
Second, as explained in \cite{PS16}, the lazy scheme of \cite{GP13} is inherently non-local.  
Since local algorithms are advantageous, \cite{PS16} also devised a local algorithm for maintaining a $(1+\eps)$-maximum matching, which is quite intricate.
An alternative simpler approach is to maintain a $(1+\eps)$-maximum matching on top of our bounded degree matching sparsifier,
by  dynamically excluding augmenting paths of length $O(1/\eps)$.
By Theorem \ref{core}, this yields an $(1+O(\eps))$-maximum matching, and the update time of the resulting algorithm is $(\alpha/\eps)^{O(1/\eps)}$, just as in \cite{PS16}.

\subsubsection{The dynamic approximate maximum IS problem} \label{ISS}
The size of the maximum IS changes by at most 1 following each update, hence we can apply the lazy scheme of \cite{GP13,PS16} also for this problem.
Notice, however, that there is no need to compute a sparsifier here, since the maximum IS is of size at least $n/(\beta+1)$ by Turan's theorem, where $\beta$ is the average degree in the graph. Hence, one can simply compute an approximate maximum IS on the entire graph, and amortize the cost of this static computation over $\Theta(\eps \cdot (n/\beta))$ update steps.
Consequently, using the lazy scheme, we get a simple reduction from the dynamic to the static case:
The update time (both amortized and worst-case) for maintaining a $(t+\eps)$-IS is smaller than the static time complexity of computing a $t$-IS by a factor of $\Omega((n / \beta) \cdot \eps)$,
for any $t \ge 1$ and $\eps > 0$. 
\vspace{8pt}
\\{\bf Remarks.}
\begin{enumerate}
\item For planar graphs, one can compute a $(1+\eps)$-IS in time $O(n)$, for any constant $\eps > 0$ \cite{Baker94}.
We can thus  maintain a $(1+\eps)$-IS with a constant  worst-case update time for any constant $\eps > 0$.
\item  For graphs with average degree bounded by $\beta$, one can compute a $(k+1)/2$-IS in time $O(n^{3/2} \cdot \beta)$ \cite{Hochbaum83}.
We can thus maintain an IS with an approximation of roughly $(k+1)/2$ and a worst-case update time of $O(\sqrt{n} \cdot \beta^{2})$.
\end{enumerate}

\subsection{Distributed networks} \label{s43}
We consider the standard $\mathcal{LOCAL}$ and $\mathcal{CONGEST}$ models of communication (cf.~\cite{PelB00}), which are standard distributed computing models capturing the essence of spatial locality and congestion.
All processors wake up simultaneously, and computation proceeds in fault-free synchronous rounds during which every processor exchanges messages of either unbounded size (in the $\mathcal{LOCAL}$ model) or of $O(\log n)$-bit size (in the $\mathcal{CONGEST}$ model).
It is easy to see that our sparsification algorithms can be implemented in distributed networks using a single communication round, even using $O(1)$-bit messages,
during which each processor sends messages along at most $\Delta$ of its adjacent edges, where $\Delta$ is the degree bound of the sparsifier.
Hence, the results mentioned in Section \ref{s14} hold w.r.t.\ both the
$\mathcal{LOCAL}$ and the $\mathcal{CONGEST}$ models of communication.
In this way we provide a clean and simple reduction from either bounded arboricity graphs (for the distributed approximate maximum matching and minimum VC problems) or from bounded average degree graphs (for the distributed approximate maximum IS problem) to bounded degree graphs.

For the distributed approximate VC problem, \cite{BCS16} showed that a $(2+\eps)$-VC can be computed in $O(\log \Delta / (\eps \log\log \Delta))$ rounds,
where $\Delta$ is the maximum degree in the graph.
We can plug our reduction to extend the result of \cite{BCS16} to graphs of arboricity bounded by $\alpha$.
\begin{theorem}
For any graph of arboricity bounded by $\alpha$ and any $\eps > 0$,
there is a distributed algorithm for computing a $(2+\eps)$-VC in $O(\log (\alpha / \eps) / (\eps \log \log (\alpha / \eps)))$ rounds.
\end{theorem}

As mentioned in Section \ref{s14}, for the distributed approximate maximum matching problem, a reduction from bounded arboricity graphs to bounded degree graphs was already given in \cite{CHS09}.
The reduction of \cite{CHS09} starts by computing carefully chosen vertex sets in the graph, using which a bounded degree subgraph is computed,
in $O(1)$ communication rounds. A distributed approximate matching algorithm is then run on that subgraph.
Based on the matching returned as output to that algorithm, another carefully chosen bounded degree subgraph is computed, in $O(1)$ more communication rounds.
A distributed approximate matching algorithm is then run on the new subgraph, and the final matching is the union of the first matching and the second.
Our reduction is obtained as an immediate corollary of our matching sparsifier from Section \ref{s31}, and it has several advantages over the one of \cite{CHS09}. First and foremost, our reduction is much simpler.
Second, it requires a single communication round, whereas the number of rounds in the reduction of \cite{CHS09} depends on the time required to compute an approximate matching.
In particular, throughout the computation of our sparsifier, the \emph{load} (as well as node congestion) on all processors is low, since any vertex sends messages only along $\Delta$ of its adjacent edges is a single round, where $\Delta$ is the degree bound of the sparsifier.
Moreover, the load on processors will remain low by definition throughout the subsequent run of any distributed approximate matching algorithm,
 since that algorithm is run on a bounded degree subgraph.
 In contrast, the computation of the subgraphs in \cite{CHS09} involves the exchange of messages between high degree vertices, and this ``high-load'' message exchange proceeds throughout $O(1)$ rounds;
then a distributed approximate matching algorithm is run on a bounded degree subgraph, which may require many rounds to complete, and later another ``high-load'' message exchange proceeds throughout $O(1)$ more rounds.
Third, the degree of our matching sparsifier is at most $O(\alpha/\eps)$, whereas the degree bound of the subgraphs of \cite{CHS09} is $O(\alpha/\eps^2)$, i.e., there is a gap of factor $1/\eps$,
and this gap may be amplified significantly in applications where the runtime of the distributed algorithm depends exponentially on the maximum degree.
In particular, \cite{EMR15} devised a distributed algorithm for computing a $(1+\eps)$-maximum matching in $\Delta^{O(1/\eps)} + O(\eps^{-2}) \cdot \log^* n$
rounds, for graphs with degree bounded by $\Delta$. Due to our matching sparsifier, we extend the result of \cite{EMR15} to graphs of arboricity bounded by $\alpha$
to get
a $(1+\eps)$-maximum matching in $(\alpha / \eps)^{O(1/\eps)} + O(\eps^{-2}) \cdot \log^* n$ rounds,
which has a better dependence on $\eps$ than if the reduction of \cite{CHS09} were to be used.





\end{document}

%% file: latexmac.tex
\def\denseformat{
\setlength{\textheight}{9in}
\setlength{\textwidth}{6.9in}
\setlength{\evensidemargin}{-0.2in}
\setlength{\oddsidemargin}{-0.2in}
\setlength{\headsep}{10pt}
\setlength{\topmargin}{-0.3in}
\setlength{\columnsep}{0.375in}
\setlength{\itemsep}{0pt}
}

\newtheorem{theorem}{Theorem}[section]

\newtheorem{claim}[theorem]{Claim}
\newtheorem{lemma}[theorem]{Lemma}

\newtheorem{corollary}[theorem]{Corollary}

\newtheorem{comment}[theorem]{Comment}

\newtheorem{observation}[theorem]{Observation}

\def\boldhead#1:{\par\vskip 7pt\noindent{\bf #1:}\hskip 10pt}
\def\ithead#1:{\par\vskip 7pt\noindent{\it #1:}\hskip 10pt}

\def\inline#1:{\par\vskip 7pt\noindent{\bf #1:}\hskip 10pt}
\def\midinline#1:{\par\noindent{\bf #1:}\hskip 10pt}
\def\dnsinline#1:{\par\vskip -7pt\noindent{\bf #1:}\hskip 10pt}
\def\ddnsinline#1:{\newline{\bf #1:}\hskip 10pt}
\def\largeinline#1:{\par\vskip 7pt\noindent{\large\bf #1:}\hskip 10pt}
\long\def\comment #1\commentend{}
\long\def\commhide #1\commhideend{}
\long\def\commfull #1\commend{#1}
\long\def\commabs #1\commenda{}
\long\def\commtim #1\commendt{#1}
\long\def\commb #1\commbend{}
\long\def\commedit #1\commeditend{} 

\long\def\commB #1\commBend{}       

\long\def\commex #1\commexend{}     

\long\def\commsiena #1\commsienaend{}  
                                         
\long\def\commBI #1\commBIend{}  
                                         
\long\def\CProof #1\CQED{}

\def\blackslug{\hbox{\hskip 1pt \vrule width 4pt height 8pt
    depth 1.5pt \hskip 1pt}}
\def\QED{\quad\blackslug\lower 8.5pt\null\par}

\long\def\PPP#1{\noindent{\bf Proof:}{ #1}{\quad\blackslug\lower 8.5pt\null}}

\long\def\denspar #1\densend
{#1}

\newif\ifnotesw\noteswtrue
   {\ifnotesw\marginpar[\hfill\(\top\)]{\(\top\)}\fi}%
   {\ifnotesw\marginpar[\hfill\(\bot\)]{\(\bot\)}\fi}

\newcommand{\mnote}[1]%
    {\ifnotesw\marginpar%
        [{\scriptsize\it\begin{minipage}[t]{\marginparwidth}
        \raggedleft#1%
                        \end{minipage}}]%
        {\scriptsize\it\begin{minipage}[t]{\marginparwidth}
        \raggedright#1%
                        \end{minipage}}%
    \fi}

\def\cM{{\cal M}}

\def\cO{{\cal O}}

\def\MathF{\hbox{\rm I\kern-2pt F}}
\def\MathP{\hbox{\rm I\kern-2pt P}}
\def\MathR{\hbox{\rm I\kern-2pt R}}
\def\MathZ{\hbox{\sf Z\kern-4pt Z}}
\def\MathN{\hbox{\rm I\kern-2pt I\kern-3.1pt N}}
\def\MathC{\hbox{\rm \kern0.7pt\raise0.8pt\hbox{\footnotesize I}
\kern-4.2pt C}}
\def\MathQ{\hbox{\rm I\kern-6pt Q}}

\newsavebox{\ttop}\newsavebox{\bbot}

\def\eps{\epsilon}

